\documentclass[reprint, aps, 10pt, pra, nofootinbib, preprintnumbers]{revtex4-2}

\usepackage{graphicx}
\usepackage{amsmath, amssymb, amsthm}
\usepackage{physics}

\newtheorem{lemma}{Lemma}
\newtheorem{proposition}{Proposition}

\usepackage{microtype}
\usepackage{hyperref}
\hypersetup{colorlinks=true}

\newcommand{\F}{\mathbb F}
\newcommand{\rk}{\operatorname{rank}}
\newcommand{\Pf}{\operatorname{Pf}}

\begin{document}

\title{R\'enyi multi-entropy of graph states
and the \texorpdfstring{$n\to1$}{n to 1} limit}

\author{Sriram Akella}
\email{sriram.akella@tifr.res.in}
\affiliation{Department of Theoretical Physics, Tata Institute of Fundamental Research, Homi Bhabha Road, Colaba, Mumbai 400005, India}

\author{Abhijit Gadde}
\email{abhijit@theory.tifr.res.in}
\affiliation{Department of Theoretical Physics, Tata Institute of Fundamental Research, Homi Bhabha Road, Colaba, Mumbai 400005, India}

\begin{abstract}
R\'enyi multi-entropy extends the R\'enyi entropy to pure states of more than two parties. It is defined by a replica contraction at integer replica number $n$, and the multi-entropy is its $n\to1$ limit. In holography this limit is conjectured to compute the area of a minimal soap-film web in the bulk. Understanding multi-entropy for general states, however, is harder. Unlike the bipartite case, the replica contraction is not a moment of a density matrix, and the continuation to $n=1$ has to be found by other means. We study this problem for qubit graph states. We show that at every odd $n$ the contraction equals $2^{-k/2}$, where $k$ is a sum of ranks of matrices over a finite field, one row and one column per qubit. We define a class of graphs that we call collision-controlled, for which the rank sum is a polynomial in $n$. This class includes graphs without even cycles and all complete bipartite graphs. The natural polynomial $n \to 1$ continuation is a finite sum over partitions of the parties. While the R\'enyi multi-entropy is non-increasing under single-qubit Pauli measurements at every odd integer $n>1$, the continuation to $n=1$ is not. A $Z$ measurement on one qubit of an odd ring with $2r + 1 \geq 7$ qubits raises it from $(r+1)\log 2$ to $(2r-1) \log 2$. Thus the continuation is not an entanglement monotone.
\end{abstract}
\preprint{TIFR/TH/26-29}
\maketitle

\section{Introduction} 

The R\'enyi entropies of a reduced density matrix $\rho_A$ are determined by the moments $\Tr\rho_A^n$, and each moment has a simple replica representation: take $n$ copies of the state and contract the indices of $A$ cyclically. The R\'enyi multi-entropy \cite{GKS} extends this construction to a pure state on $q$ parties. The copies are placed on a periodic $(q-1)$-dimensional lattice with $n$ sites in each direction, and each of the first $q-1$ parties is contracted cyclically along its lattice direction. This gives a party-symmetric local-unitary invariant built from $n^{q-1}$ copies. The multi-entropy is the limit $n \to 1$. At $q=2$, it reduces to the usual entanglement entropy. 

For holographic states, the $n \to 1$ limit has a simple geometric description. At leading semiclassical order the R\'enyi multi-entropy is conjectured to be the area of a web of branes anchored on the boundary partition \cite{GKS}, in analogy with the cosmic branes that compute holographic R\'enyi entropies \cite{Lewkowycz:2013nqa, Dong:2016fnf}. As $n\to1$, the branes become light and no longer backreact on the bulk geometry. The limiting web is a minimal soap-film configuration, whose area is conjectured to compute the multi-entropy, in analogy with the Ryu--Takayanagi prescription for entanglement entropy~\cite{Ryu:2006bv, Ryu:2006ef}. The proposal has since been tested in two-dimensional conformal field theory \cite{HarperTakayanagiTsuda} and in random tensor networks \cite{Penington:2022dhr, Akella:2026bci, Hu:2026bhg}, extended to mixed states through a reflected construction \cite{YuanLiZhou}, placed within a larger family of replica contractions defined by permutations \cite{GKSClassification}, and studied under coarse graining \cite{GaddeMonotonicity}.

In holography, the bulk geometry gives a natural $n \to 1$ limit, and the resulting soap-film area suggests that the multi-entropy captures something genuine about multipartite entanglement \cite{Walter:2016lgl, Horodecki:2024bgc}. To identify what, we want to take the $n \to 1$ limit for a general multipartite state. This is, however, much harder than in the bipartite case. For two parties, the replica invariant is the moment of a positive operator, whose spectrum gives the continuation to real $n$. For $q > 2$, no such universal continuation prescription is known~\cite{GKS,GKSClassification}. We ask what the integer replica data alone can determine, and what additional input is required for a general definition of the multi-entropy.

Qubit graph states \cite{Hein:2006uvf, Schlingemann:2000eqc} make the question tractable. They are stabilizer states \cite{Gottesman}, meaning each is the unique joint eigenstate of a maximal set of commuting Pauli operators, and they are labeled by simple graphs; they include the cluster states of measurement-based quantum computation \cite{RaussendorfBriegel, Briegel:2000unx}, and their entanglement properties translate into properties of the graph \cite{HeinEisert, Fattal:2004frh}. Local Clifford gates, which map Pauli operators to Pauli operators under conjugation, and local Pauli measurements take the form of elementary graph operations~\cite{VandenNest}. For these states the replica contraction can be carried out exactly at every odd $n$. It reduces to linear algebra over the finite field $\mathbb{F}_2$, and the whole computation becomes a sum of ranks of matrices with one row and one column per qubit \cite{Akella:2026xza, Czech:2026cxw}.

Our results are as follows. For a class of graphs that we call \emph{collision-controlled}, and for odd $n$, the rank sum $r_G(n)$ that determines the R\'enyi multi-entropy is a polynomial in $n$. The collision-controlled class includes every graph without an even cycle, and all complete bipartite graphs. The $n \to 1$ polynomial continuation is a finite sum over partitions of the party set $\{1, \dots, q\}$. Some results of the analytic continuation are shown in Table~\ref{tab:results}. Our second result is that the R\'enyi multi-entropy is non-increasing under single-qubit Pauli measurements at every odd integer $n>1$, whereas its polynomial continuation to $n=1$ need not be. This implies, in particular, that the continuation to $n=1$ is not an entanglement monotone \cite{Vedral:1997qn, Vidal, Plenio:2007zz, Horodecki:2009zz}. For example, a $Z$ measurement on one qubit of a seven-vertex ring lowers the R\'enyi multi-entropy at every odd integer $n>1$, but raises the value of the $n=1$ continuation (see Fig.~\ref{fig:C7_to_P6}). Both the ring and the resulting path are collision-controlled, so the failure of monotonicity occurs within the class where the entire odd-$n$ sequence has a unique polynomial continuation. 

\begin{figure}
    \centering
    \includegraphics[width=\linewidth]{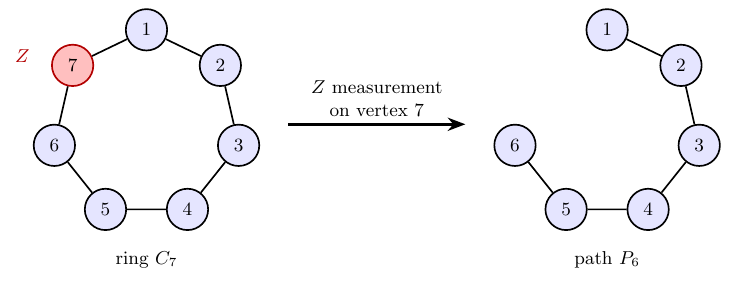}
    \caption{A $Z$ measurement on one qubit of the seven-vertex ring $C_7$ results in the six-vertex path $P_6$. The measurement lowers the R\'enyi multi-entropy at every odd integer $n>1$, but raises the $n \to 1$ multi-entropy. Both graphs are collision-controlled, so the odd-$n$ R\'enyi multi-entropy has a unique polynomial continuation.}
    \label{fig:C7_to_P6}
\end{figure}

\begin{table}
\caption{Values $S_1$ obtained by polynomial continuation to $n=1$, with one qubit per party. The derivations are given in Appendix~\ref{app:collision-examples}.}
\label{tab:results}
\begin{ruledtabular}
\begin{tabular}{lc}
graph & $S_1/\log2$\\
\hline
forest with $m$ edges & $m$\\
odd cycle on $2r+1$ vertices & $r+1$\\
complete bipartite, class sizes $m_1,m_2$ & $m_1+m_2-1$
\end{tabular}
\end{ruledtabular}
\end{table}

In section \ref{sec:ranksum}, show how the R\'enyi multi-entropy reduces to a rank problem over the finite field $\mathbb{F}_2$. In section \ref{sec:collision-controlled} we define collision-controlled graphs, and show how the rank problem simplifies to give a formula for the R\'enyi multi-entropy for odd $n$. Then we study the behavior of R\'enyi multi-entropy for odd $n$ and also in the $n\to1$ limit. In section \ref{sec:monotonicity}, we show the natural $n \to 1$ limit is not monotonic under local measurements. The rank proofs and the calculations for forests, odd cycles, and complete bipartite graphs are collected in the two appendices \ref{app:rank-reduction} and \ref{app:collision-examples}.

\section{R\'enyi multi-entropy as a rank sum}
\label{sec:ranksum}

Let $|\psi\rangle\in\mathcal H_1\otimes\cdots\otimes\mathcal H_q$ be a normalized pure state on $q\geq2$ parties. For a positive integer $n$ we take $n^{q-1}$ copies of $|\psi\rangle \langle \psi |$ and label them by points of the periodic lattice $\Lambda_n=\mathbb Z_n^{\,q-1}$. Let $e_a$ be the unit vector in the $a^{\text{th}}$  lattice direction for $a=1,\ldots,q-1$, and set $e_q=0$. On the copies of party $a$, let $\sigma_a$ be the permutation that sends the copy at $s$ to the copy at $s+e_a$; in particular $\sigma_q$ is the identity. The replica partition function is~\cite{GKS}
\begin{equation}
Z_n^{(q)}(\psi)=\langle\psi|^{\otimes n^{q-1}}\bigl(\sigma_1\otimes\cdots\otimes\sigma_q\bigr)|\psi\rangle^{\otimes n^{q-1}}.
\label{eq:replica-invariant}
\end{equation}
Figure~\ref{fig:replica-lattice} shows the case $q=n=3$. Singling out party $q$ as the one that is not shifted removes a common translation of all copies; the result is nevertheless a local unitary invariant that is symmetric under permutations of all $q$ parties. For $q=2$ the lattice is $\mathbb Z_n$, $\sigma_1$ is an $n$-cycle, and $Z_n^{(2)}=\Tr\rho_1^n$, where $\rho_1$ is the reduced density matrix of the first party.

For $n>1$, we define the R\'enyi multi-entropy as
\begin{equation}
S_n^{(q)}(\psi)=\frac{1}{1-n}\log Z_n^{(q)}(\psi).
\label{eq:renyi-multientropy}
\end{equation}
Ref. \cite{GKS} includes an extra factor $n^{-(q-2)}$. It tends to one as $n\to1$, so the two conventions have the same multi-entropy whenever the limit exists, and we drop it because it clutters the rank formulas below.

\begin{figure}
\centering
\includegraphics[width=0.55\linewidth]{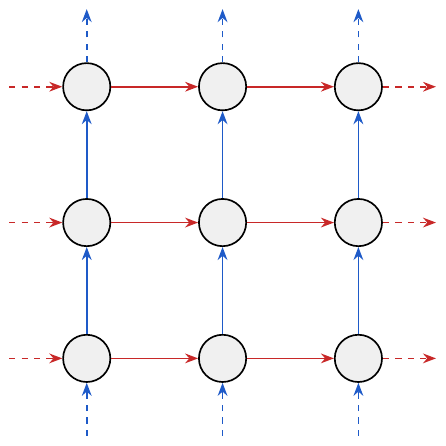}
\caption{Replica lattice for $q=3$ and $n=3$. Each site carries one copy of the state. The permutations of the first and second parties translate the copy label along the horizontal and vertical directions, while the third party is not shifted. Dashed arrows indicate the periodic identifications.}
\label{fig:replica-lattice}
\end{figure}

\textit{Qubit graph states.} Let $G=(V,E)$ be a simple graph with $N$ vertices and binary adjacency matrix $\gamma$. Each vertex carries a qubit, and a map $\alpha:V\to\{1,\ldots,q\}$ assigns the qubit at $u$ to party $\alpha(u)$. Let $|+\rangle\equiv(|0\rangle+|1\rangle)/\sqrt2$. The graph state $|G\rangle$ is obtained from $|+\rangle^{\otimes N}$ by applying a controlled-$Z$ gate along every edge, i.e., 
\begin{equation}\label{eq:ketG}
    \ket{G} = \left(\prod_{\{u, v\} \in E} CZ_{uv}\right) \ket{+}^{\otimes N}. 
\end{equation}
The controlled-$Z$ gates commute with each other, so the state does not depend on the order in which these gates are applied.
Any stabilizer state is local-unitary equivalent to a graph state \cite{Schlingemann:2001zyo, VandenNest, Grassl:2007hcz}. Therefore, our results hold more generally for stabilizer states. 

Given the graph $G$, let us define the quadratic form
\begin{equation}
    \Phi_{G}(\mathbf{x}) = \sum_{u < v} \gamma_{uv} x_u x_v,
\end{equation}
where $\mathbf{x}$ is a binary vector in $\mathbb{F}_2^N$. Substituting the graph state from Eq.~\eqref{eq:ketG} in the expression for the replica partition function in Eq.~\eqref{eq:replica-invariant}, we get \cite{Akella:2026xza, Czech:2026cxw} 
\begin{equation}\label{eq:gauss-sum}
    Z_{n}^{(q)} = \langle + \dots + | \mathcal{G}\rangle = \frac{1}{2^d} \sum_{\mathbf{x} \in \mathbb{F}_2^d} (-1)^{\Phi_\mathcal{G}(\mathbf{x})},
\end{equation}
where $\mathcal{G}$ is the ``big graph'' state obtained from $G$ over $d = N n^{q-1}$ vertices. These vertices are labeled by the tuple $(a, u)$ where $a$ is a replica label running from $1$ to $n^{q-1}$ and $u$ is the vertex label running from $1$ to $N$. Introduce the shorthand $\sigma_u$ for the permutation $\sigma_{\alpha(u)}$ corresponding to the party $\alpha(u)$ that the vertex $u$ belongs to. The adjacency matrix for the big graph $\mathcal{G}$ is
\begin{equation}
    \Gamma_{(a, u), (b, v)} = \gamma_{uv}\left(\delta_{ab} + \delta_{\sigma_u(a) \sigma_v(b)}\right) \mod 2. 
\end{equation}
See Fig.~\ref{fig:big_q3_n2} for an example of the big graph $\mathcal{G}$ constructed from the triangle graph and the $n = 2$, $q = 3$ R\'enyi multi-entropy.  

\begin{figure}
    \centering
    \includegraphics[width=0.75\linewidth]{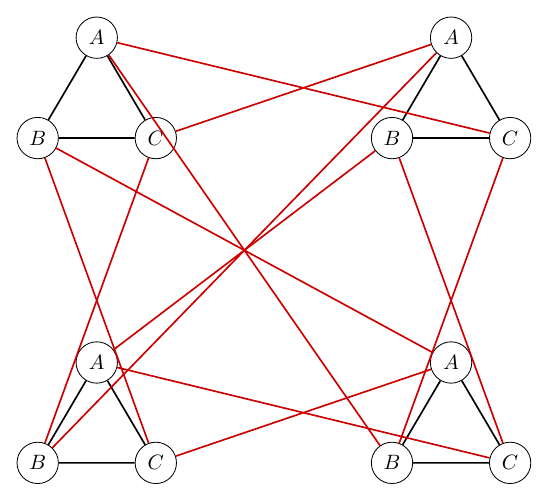}
    \caption{The big graph $\mathcal{G}$ for $q= 3$ and $n= 2$, built from the triangle with its vertices $A, B, C$ in three different parties. Each of the $n^{q-1} = 4$  copies carries a copy of $G$ (black edges), from the term $\delta_{ab}$ in $\Gamma$. The red edges come from the term $\delta_{\sigma_u(a) \sigma_v(b)}$. They join vertices of different copies whose bra indices are contracted by the permutations of the R\'enyi multi-entropy.}
    \label{fig:big_q3_n2}
\end{figure}

\textit{The rank formula.} The Gauss sum in Eq.~\eqref{eq:gauss-sum} vanishes if $\Phi_{\mathcal G}$ is not identically zero on $\ker \Gamma$, and can pick up a negative sign depending on the Arf invariant \cite{Czech:2026cxw, Arf}.
Let \(r_G(n)\) be the rank over $\mathbb F_2$ of the adjacency matrix $\Gamma$ of the big graph $\mathcal{G}$. For odd \(n\), we show in Appendix~\ref{app:rank-reduction} that the Gauss sum neither vanishes nor acquires a negative sign, giving
\begin{equation}
Z_n^{(q)}(G)=2^{-r_G(n)/2},\quad
S_n^{(q)}(G)=\frac{\log2}{2(n-1)}\,r_G(n).
\label{eq:rank-partition-function}
\end{equation}
For even $n$, we do not have a proof, although numerical checks support the same result.

The idea is that the permutations defining the R\'enyi multi-entropy can simultaneously be diagonalized in the \emph{replica momentum} basis, and the eigenvalues are $n^{\text{th}}$ roots of unity. For odd $n$, these roots can be embedded in a field extension $\mathbb{F}_{2^m}$ of $\mathbb{F}_2$, where $m$ is such that $n$ divides $2^m - 1$. For even $n$, this is not possible, and we stick to odd $n$ for the rest of the paper. 

A simple check: for two parties, we have $r_G(n)=2(n-1)\rk_{\F_2}\gamma_{12}$, where $\gamma_{12}$ is the block of $\gamma$ joining the two parties, and Eq.~\eqref{eq:rank-partition-function} gives $S_n^{(2)}=\rk_{\F_2}(\gamma_{12})\log2$ for every odd $n$, the familiar flat spectrum of a bipartite graph state. 

The $n \to 1$ limit, therefore, depends on the rank $r_G(n)$ of the adjacency matrix $\Gamma$. This is a large matrix, but it becomes simple in the replica momentum basis. A joint eigenspace of the R\'enyi multi-entropy translations is labeled by assigning one root of unity to each party: $t = (t_1, \dots, t_{q-1})$, with $t_q = 1$ for the reference party. These roots are to be thought of as elements of the finite field $\mathbb{F}_{2^m}$. On this eigenspace, the adjacency matrix $\Gamma$ becomes (up to a diagonal rescaling) the $N \times N$ matrix 
\begin{equation}\label{eq:Upsilon-t}
    \left(\Upsilon(t)\right)_{uv} = \gamma_{uv}(t_{\alpha(u)} - t_{\alpha(v)}).  
\end{equation}
It is the adjacency matrix of $G$ with edge $\{u,v\}$ weighted by the difference of the replica momenta of its endpoints. It is alternating and also symmetric in characteristic two.

For every odd $n$, the $\mathbb{F}_2$ rank of $\Gamma$ is simply the sum of ranks of $\Upsilon(t)$ over all collective momentum labels $t$: 
\begin{equation} \label{eq:odd-replica-rank-sum}
    r_G(n) = \sum_{t} \rk \Upsilon(t).
\end{equation}
The dependence on the R\'enyi index $n$ only appears in the replica momenta $t$. We prove Eq.~\eqref{eq:odd-replica-rank-sum} in Appendix~\ref{app:rank-reduction}.

\begin{figure}
    \centering
    \includegraphics[width=\linewidth]{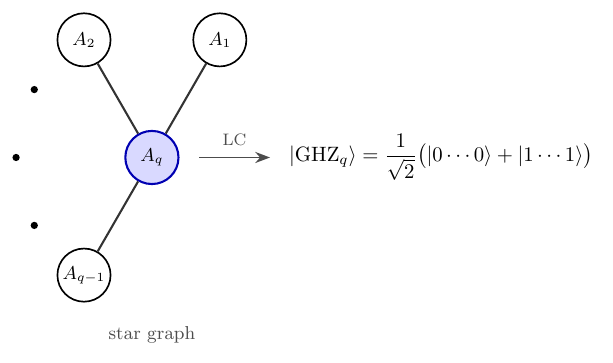}
    \caption{The star graph with center $A_q$ and leaves $A_1, \dots, A_{q-1}$, each vertex in its own party. A Hadamard gate on every leaf, which is a local Clifford (LC) operation, maps the star graph state to the $q$-party GHZ state. The two states therefore have the same R\'enyi multi-entropy.}
    \label{fig:ghz-star-graph}
\end{figure}

Let us test this formula in a simple case. 
The $q$-party GHZ state is locally unitarily equivalent to the graph state of the star graph (see Fig.~\ref{fig:ghz-star-graph}), with the center in one party and one leaf in each of the others~\cite{Hein:2006uvf}. 
Choose the center as the reference party, so that $t_q=1$, and define $w=(t_1-1,\ldots,t_{q-1}-1)^{\mathsf T}$. Ordering the leaves first and the center last, the weighted adjacency matrix is
\[
\Upsilon(t)=
\begin{pmatrix}
0_{(q-1)\times(q-1)} & w\\
w^{\mathsf T} & 0
\end{pmatrix}.
\]
Every leaf column has entries only in the center row, so the leaf columns are all proportional. If $w\ne0$, at least one leaf column is nonzero, and the center column $(w^{\mathsf T},0)^{\mathsf T}$ is independent of it because its nonzero entries lie in the leaf rows. The rank is therefore two. If $w=0$, the entire matrix vanishes. For odd $n$, each of the $q-1$ leaf momenta has $n$ possible values, giving $n^{q-1}$ assignments. Exactly one, with every $t_i=1$, gives rank zero; all others give rank two. Hence
\begin{equation}
r_{\mathrm{GHZ}}(n)=2\bigl(n^{q-1}-1\bigr),\quad
S_n^{(q)}(\mathrm{GHZ})=\frac{n^{q-1}-1}{n-1}\log2.
\label{eq:ghz}
\end{equation}
This agrees with a direct evaluation of the R\'enyi multi-entropy of the GHZ state. The ket copies must be all-zero or all-one, and the permutations force every copy to take the same value, so $Z_n^{(q)}=2\cdot2^{-n^{q-1}}$.

For the star graph, the rank of $\Upsilon(t)$ depends only on which replica momenta coincide, and not on their particular values. Counting these momenta assignments gives a polynomial in $n$ that reproduces the rank sum at every odd integer. 
We turn this property of the GHZ state into the definition of \emph{collision-controlled} graphs. The rank of $\Upsilon(t)$ does not depend on the individual momenta of $t$, but rather only on the pattern of equalities or ``collisions'' in $t$. These patterns can then be counted using simple combinatorics, and the result can be analytically continued to $n \to 1$. 

\section{Collision-controlled graphs}
\label{sec:collision-controlled}

The equalities among $t_1,\ldots,t_q$ divide the parties into groups. We call this partition $\pi(t)$ the collision partition of $t$. Let $\Pi_q$ be the set of all partitions of the $q$ parties, and $|\pi|$ be the number of groups in a given partition $\pi$. For each $\pi\in\Pi_q$, let $G_\pi$ be the graph obtained from $G$ by deleting every edge whose endpoints belong to parties in the same group. When $\pi(t)=\pi$, the nonzero entries of $\Upsilon(t)$ are exactly the edges of $G_\pi$.

For a fixed assignment of vertices to parties, we say that $G$ is collision-controlled if the rank of $\Upsilon(t)$ depends only on the collision partition $\pi$ of $t$ for every assignment of replica momenta $t$, over every finite extension of $\F_2$. We denote this rank by $\varrho_G(\pi)$: for a fixed graph $G$, it depends on the momentum assignment $t$ only through its collision partition $\pi$.

The number of replica momenta with a given collision partition is easy to count. If $\pi$ has $k$ groups, the group containing the reference party has replica momentum one, and the remaining $k-1$ groups have distinct replica momenta picked from the other $n-1$ roots of unity. This can be done in $(n-1)_{k-1}$ ways, where $(x)_0=1$ and $(x)_k=x(x-1)\cdots(x-k+1)$ is the falling factorial. Summing over partitions gives the following result:
\begin{equation}
    r_G(n) = P_G(n), \quad P_G(x) = \sum_{\pi \in \Pi_q} (x-1)_{|\pi|-1} \varrho_G(\pi). 
\end{equation}
The polynomial $P_G$ satisfies $P_G(1)=0$, so Eq.~\eqref{eq:rank-partition-function} has a regular limit at $n=1$. 
Since the derivative of $(x-1)_{k-1}$ at $x=1$ is $(-1)^{k}(k-2)!$ for $k\geq2$, the multi-entropy for a collision-controlled graph state is
\begin{equation}
S_1^{(q)}(G)=\frac{\log2}{2}\sum_{\substack{\pi\in\Pi_q\\|\pi|\geq2}}(-1)^{|\pi|}(|\pi|-2)!\;\varrho_G(\pi).
\label{eq:collision-entropy}
\end{equation}
For the GHZ state, $\varrho_G(\pi)=2$ unless $\pi$ has a single group, and Eq.~\eqref{eq:collision-entropy} reproduces $(q-1)\log2$.

The polynomial $P_G$ is the unique analytic continuation of the odd-replica ranks within the following growth class. Suppose $f(n)=r_G(n)$ at every odd integer $n\geq1$, and $f(z)$ is analytic on a neighborhood of $\operatorname{Re}z\geq1$, with $|f(z)|\leq Ce^{c|z|}$ there and $|f(1+iy)|\leq Ce^{\tau|y|}$ for real $y$, where $C,c$ are finite positive constants and $0\leq\tau<\pi/2$. Setting $z=2w+1$ and applying Carlson's theorem~\cite{Titchmarsh} to $f(2w+1)-P_G(2w+1)$ gives $f=P_G$.

We now ask which graphs are collision-controlled. Fix $t$ and let $\pi=\pi(t)$. The nonzero entries of $\Upsilon(t)$ sit exactly on the edges of $G_{\pi}$. The idea is to now bound the rank of $\Upsilon(t)$. A fundamental result in graph theory \cite{Tutte} says
\begin{equation}
    \rk \Upsilon(t) \leq 2 \nu(G_{\pi}),
\end{equation}
where $\nu(G_{\pi})$ is the number of edges in a maximum matching of $G_{\pi}$. We give a short proof of this bound below. Recall that a matching of a graph is a collection of its edges with no common vertex. A maximum matching is a matching with largest number of edges. 

If our graph is such that the bound is always saturated, then the rank depends only on the collision partition $\pi$, so the graph is collision-controlled. We therefore ask when the bound is saturated.

To answer this, let us also bound the rank of $\Upsilon(t)$ from below. Consider a subset of vertices $U$ of $G_{\pi}$ and the induced weighted adjacency matrix $A(t)$ on the vertices $U$. If $A(t)$ is nonsingular, i.e., $\det A(t) \neq 0$, then we have 
\begin{equation}
    |U| \leq \rk \Upsilon(t). 
\end{equation}

Let $U$ be the set of vertices covered by a maximum matching $M$ of $G_{\pi}$. Then $|U| = 2 \nu(G_{\pi})$. To apply the above bound, however, we need $A(t)$ to be nonsingular. The Pfaffian of $A(t)$ is a sum over perfect matchings of the induced subgraph $H_{\pi}$ on the vertices $U$: 
\begin{equation}
    \Pf A(t) = \sum_{M} \prod_{\{u, v\} \in M} A(t)_{uv}.
\end{equation}
The Pfaffian satisfies $\det A(t)=(\Pf A(t))^2$. This expansion also explains the upper bound on $\rk\Upsilon(t)$. If $\Upsilon(t)$ has rank $2k$, it has a nonsingular principal submatrix of size $2k$, obtained by choosing $2k$ vertices and retaining corresponding rows and columns. The Pfaffian of this submatrix is nonzero, so its expansion contains at least one nonzero term, corresponding to a perfect matching on these $2k$ vertices. Thus $G_\pi$ contains a matching of $k$ edges, implying $k\leq\nu(G_\pi)$ and hence $\rk\Upsilon(t)\leq2\nu(G_\pi)$.

Recall that a perfect matching covers all vertices of the graph. If $H_{\pi}$ is such that $M$ is the only perfect matching, then since $A(t)_{uv}\ne0$ for every edge $\{u,v\}$ of $H_{\pi}$, we get $\Pf A(t) \neq 0$, and 
\begin{equation}
    \rk \Upsilon(t) = 2 \nu(G_{\pi}). 
\end{equation}
An example is shown in Fig.~\ref{fig:matching}. 

\begin{figure}
    \centering
    \includegraphics[width=\linewidth]{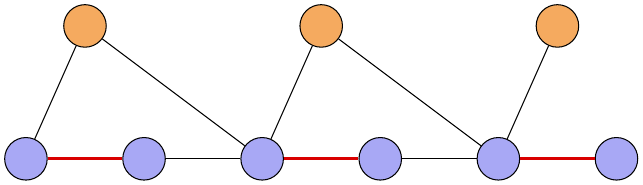}
    \caption{A maximum matching $M$ (red edges) that covers the blue vertices $U$. Numbering the blue vertices from the left, every edge touches the first, third, or fifth one, so no matching has more than three edges and $\nu = 3$. The subgraph induced on $U$ is a path, and $M$ is its only perfect matching. The Pfaffian of $A(t)$ is therefore the product of the three red edge weights. It is nonzero, so $\rk \Upsilon(t) = |U|= 2\nu = 6$.}
    \label{fig:matching}
\end{figure}

We call $G$ \emph{matching-controlled} if, for every collision partition $\pi$, the graph $G_\pi$ has a maximum matching that is the only perfect matching of the subgraph induced on the vertices it covers.

\begin{proposition}
\label{prop:matching}
Every matching-controlled graph is collision-controlled, with $\varrho_G(\pi)=2\nu(G_\pi)$ for every collision partition $\pi$.
\end{proposition}

The discussion so far might suggest matching-controlled graphs are very special and complicated. However, we have the following:

\begin{proposition}
\label{lem:no-even-cycle}
If $G$ contains no even cycle, then $G$ is matching-controlled.
\end{proposition}

\begin{proof}
Each $G_\pi$ is obtained by deleting edges, so it has no even cycle either. If a maximum matching $M$ of $G_\pi$ is not a unique perfect matching of the subgraph induced on its vertices, there is a second perfect matching $M'$ on the same vertices, and the symmetric difference of $M$ and $M'$ would contain a cycle alternating between the two, necessarily of even length.
\end{proof}

This covers all forests and all graphs whose cycles are odd. For a forest with one qubit per party, Eq.~\eqref{eq:collision-entropy} gives one unit of $\log2$ per edge, and for an odd cycle $C_{2r+1}$ it gives $(r+1)\log2$. We derive both values in Appendix~\ref{app:collision-examples}. 

\subsection{Local complementation} 

Local complementation at a vertex $v$ toggles the edges between its neighbors, adding absent edges and removing present ones, while leaving all other edges unchanged. Denote the resulting graph by $LC_v(G)$, and write $\Upsilon_G(t)$ for the weighted adjacency matrix associated with $G$. Local complementation is implemented by local Clifford unitaries on graph states~\cite{VandenNest}. 

\begin{proposition}
\label{prop:lc-invariance}
If $G$ is collision-controlled and $H=LC_v(G)$, with the assignment of vertices to parties unchanged, then $H$ is collision-controlled with $\varrho_H(\pi)=\varrho_G(\pi)$ for every collision partition $\pi$.
\end{proposition}
In fact, local complementation preserves the weighted rank even without assuming collision control. For any graph $G$ and $H=LC_v(G)$, with the same party assignment,
\[
\rk\Upsilon_H(t)=\rk\Upsilon_G(t)
\]
for every momentum assignment over every finite extension of $\mathbb{F}_2$. We prove this stronger statement in Appendix~\ref{app:rank-reduction}. Thus every graph obtained from a matching-controlled graph by local complementations is collision-controlled and has the same polynomially continued multi-entropy, even if it is not itself matching-controlled.

The four-cycle illustrates why matching control need not be preserved. It is related by local complementations to the four-vertex path, which is matching-controlled, but its two perfect matchings prevent it from being matching-controlled itself. Proposition~\ref{prop:lc-invariance} nevertheless guarantees collision control. We can also verify this directly. Taking the edge weights $w_{uv}=t_u-t_v$, its two perfect matchings contribute
\begin{equation}
\Pf \Upsilon_{C_4}(t)=w_{12}w_{34}+w_{23}w_{41}=(t_1-t_3)(t_2-t_4),
\label{eq:c4-pfaffian-factorization}
\end{equation}
so the rank drops only when opposite momenta coincide. Thus $C_4$ is collision-controlled even though it has two perfect matchings.

More generally, let $K_{m_1,m_2}$ have vertex classes of sizes $m_1$ and $m_2$, with every edge between them, and one vertex in each party. Denote the momenta on the two classes by $u_i$ and $v_j$, with $v_{m_2}=1$ as the reference. The weighted adjacency matrix is
\begin{equation}
\Upsilon =\begin{pmatrix}0&L\\L^{\mathsf T}&0\end{pmatrix},\qquad L_{ij}=u_i+v_j,
\end{equation}
so $\rk \Upsilon =2\rk L$. 
Let $u=(u_1,\ldots,u_{m_1})^{\mathsf T}$ and $v=(v_1,\ldots,v_{m_2})^{\mathsf T}$, and let $\mathbf1_k$ denote the column vector of $k$ ones. Writing $L=u\mathbf1_{m_2}^{\mathsf T}+\mathbf1_{m_1}v^{\mathsf T}$ shows that $\rk L\leq2$, with $\rk L=2$ unless all $u_i$ are equal or all $v_j$ are equal, and $\rk L=0$ only when all momenta coincide. The rank is therefore controlled by collisions. The two degenerate sets contain $n^{m_2}$ and $n^{m_1}$ momenta and intersect in $n$ of them, which gives, at every odd $n$,
\begin{equation}
P_{K_{m_1,m_2}}(n)=4n^{m_1+m_2-1}-2n^{m_1}-2n^{m_2}+2n-2
\label{eq:complete-bipartite}
\end{equation}
and $S_1(K_{m_1,m_2})=(m_1+m_2-1)\log2$. Here and below we suppress the superscript $(q)$ when every vertex is its own party. 
Collision control can again be understood through Proposition~\ref{prop:lc-invariance}. Choose an edge $ab$, with $a$ in the first vertex class and $b$ in the second. Local complementations successively at $a,b,a$ transform $K_{m_1,m_2}$ into a tree with two adjacent centers, with $m_1-1$ leaves attached to $a$ and $m_2-1$ leaves attached to $b$. This tree is matching-controlled, so its LC-equivalent graph $K_{m_1,m_2}$ is collision-controlled.

We now ask whether this polynomial continuation behaves like an entanglement measure.

\section{Monotonicity under local measurements} 
\label{sec:monotonicity}

A measure of entanglement should not increase, on average, under local operations. For graph states the natural test is a local Pauli measurement, which acts on the graph by a local Clifford transformation followed by deletion of the measured vertex \cite{VandenNest, HeinEisert, RaussendorfBriegel}. 
By Proposition~\ref{prop:lc-invariance}, the local complementations leave every weighted rank unchanged, so it suffices to study vertex deletion.

Deleting a vertex $v$ removes the row and the column of $\Upsilon(t)$ that belong to $v$, resulting in $\Upsilon'(t)$. The resulting matrices are principal submatrices, and a principal submatrix cannot have larger rank. Therefore,
\begin{equation}
\rk \Upsilon(t) \geq \rk \Upsilon'(t) 
\label{eq:pointwise-deletion-rank}
\end{equation}
for every assignment of replica momenta. Summing over $t$ we conclude that the R\'enyi multi-entropy, for any odd $n > 1$, does not increase under a local Pauli measurement, i.e., 
\begin{equation}
    S_n^{(q)}(G \setminus v)\leq S_n^{(q)}(G).
\end{equation}
One bookkeeping remark is needed before comparing continued values. Choose an unmeasured party as the reference. If the measured vertex was alone in its party, the replica momentum of that party no longer appears in $\Upsilon'(t)$, and summing over it multiplies the rank sum by $n$. Since the polynomial vanishes at $n=1$, this factor does not change its derivative there, so keeping or discarding the empty party gives the same continued value.

The inequality holds at each odd $n$ separately, but the continued value is a signed combination of the ranks, Eq.~\eqref{eq:collision-entropy}, and deleting a vertex can remove negative terms as well as positive ones. A $Z$ measurement at any vertex of the seven-cycle $C_7$ leaves the six-vertex path $P_6$ (see Fig.~\ref{fig:C7_to_P6}), and their multi-entropies are $S_1(C_7)=4\log2$ and $S_1(P_6)=5\log2$. This demonstrates that the $n \to 1$ limit is not monotonic even though the R\'enyi versions are monotonic at every odd integer $n>1$. 

For instance, at $n=3$, we have $r_{C_7}(3)=3976$, while the path with an empty seventh party has rank sum $3\,r_{P_6}(3)=3180$. The reversal is produced entirely by the continuation. More generally, deleting a vertex of $C_{2r+1}$ leaves the path $P_{2r}$ on $2r$ vertices, and $S_1(C_{2r+1})=(r+1)\log2$, $S_1(P_{2r})=(2r-1)\log2$.
The continued value increases for every $r\geq3$.

All outcomes of a Pauli measurement on a graph state differ by local Clifford corrections and have the same continued value, so the increase survives averaging over outcomes. The polynomial continuation is not an entanglement monotone under local operations and classical communication~\cite{Vidal}. Monotonicity at every odd integer $n>1$, therefore, does not imply monotonicity of the continued value at $n=1$.

\textit{Discussion.} For qubit graph states, the replica contraction at every odd replica number is $2^{-r_G(n)/2}$, where $r_G(n)$ is a sum of ranks of small matrices over a finite field. When these ranks  depend only on which replica momenta coincide, the odd-$n$ sequence has a unique polynomial continuation. This class contains all graphs without even cycles and all complete bipartite graphs. Since $r_G(n)$ is a local-unitary invariant, the same polynomial continuation also holds for every graph locally equivalent to one of these. Even within this class, the multi-entropy can increase under a local Pauli measurement, although the odd-$n$ value is non-increasing. The multi-entropy in Eq.~\eqref{eq:collision-entropy} weighs the collision partitions with alternating signs, so the pointwise rank inequality does not carry over to the derivative at $n=1$. For $C_7$ and $P_6$, the exact ranks at $n = 3$ are monotonic, $3976 > 3180$, yet the continued values are $4 \log 2 < 5 \log 2$. Every odd cycle with seven or more vertices gives such a violation. 

Our proofs use odd $n$, but numerically the even-$n$ values behave the same way. In the examples checked at $n=2,4,6$, the exact rank $r_G(n)$ agrees with the odd-$n$ polynomial, and the Gauss sums are positive and nonzero. If this holds in general, the polynomial reproduces the data at every integer $n \geq 2$. Carlson’s theorem then fixes the continuation within its
usual growth class.

If monotonicity under local operations is required, the $n\to1$ prescription must therefore differ from this polynomial continuation. By Carlson's theorem, any alternative analytic continuation of the same odd-replica ranks must fall outside the specified growth class. This leaves open which analytic or physical conditions should select the limit for a general state. In holography the bulk geometry supplies a prescription through the soap-film area. Identifying its counterpart for general quantum states would clarify which multipartite property that area computes.

For graph states of qudits with prime local dimension \(p\), the replica contraction is a quadratic Gauss sum over the finite field \(\mathbb F_p\). The replica translations can again be diagonalized when \(n\) is not divisible by \(p\), reducing the calculation to small finite-field matrices. To obtain an entropy–rank formula, we must also determine the phase of the Gauss sum. It remains to formulate the appropriate notion of collision control for qudit graph states and determine when it gives a polynomial continuation.

\textit{Acknowledgments.} This work is supported by the Department of Atomic Energy, Government of India, under Project Identification No. RTI-4012, and the Infosys Endowment for the study of the Quantum Structure of Spacetime.

\textit{AI Usage Statement.} We used Claude Opus 5.5 (Anthropic) and GPT-6 Astra (OpenAI) in the preparation of this paper. \textbf{Drafting.} The appendices were first drafted using LLMs and later edited by the authors. The main text was written by the authors. \textbf{Figures.} The TikZ code for all the figures was LLM-generated. \textbf{Language.} LLMs were used to correct typos, ensure correct grammar usage, and improve the clarity of some sentences in the main text. \textbf{Error-checking.} LLMs were used to check for errors and gaps in the argument. \textbf{Research Content.} The research question and the main results are the authors' own. The definition of collision-controlled graphs, the matching/Pfaffian criterion, and the proof of the positivity of the Gauss sum in Appendix \ref{app:A3} were developed with LLM assistance. No AI tool is an author, and the authors take full responsibility for the content of the paper.

\bibliography{references}

\appendix

\section{Proofs for odd \texorpdfstring{$n$}{n}}
\label{app:rank-reduction}

In this appendix we show that, for odd $n$, the replica partition function of a graph state is $Z_n^{(q)}=2^{-r_G(n)/2}$, where $r_G(n)=\sum_t\rk\Upsilon(t)$ [Eqs.~\eqref{eq:rank-partition-function} and~\eqref{eq:odd-replica-rank-sum}]. We also prove that local complementation does not change the ranks of $\Upsilon(t)$ (Proposition~\ref{prop:lc-invariance}). We sketch the arguments rather than write them out in full. The standard facts about finite fields that we use can be found in Ref.~\cite{LidlNiederreiter}.

\paragraph*{The idea.}
The big graph $\mathcal G$ is made of $n^{q-1}$ identical copies of $G$, one at each site of the replica lattice. If we move every copy by the same lattice vector, $\mathcal G$ does not change. The calculation of $Z_n^{(q)}$ is therefore a translation-invariant quadratic problem on a lattice, like a free field, and we solve it in the same way: in momentum space. Three familiar facts carry over.
\begin{enumerate}
\item \emph{Momentum is conserved.} The adjacency matrix $\Gamma$ of $\mathcal G$ does not mix different replica momenta $t$. It splits into small $N\times N$ blocks, one for each $t$, and the block is $\Upsilon(t)$. The rank of $\Gamma$ is then the sum of the ranks of the blocks. This is the rank sum.
\item \emph{The modes decouple.} $Z_n^{(q)}$ is an average of $(-1)^{\Phi_{\mathcal G}}$ over bit strings, which is a discrete version of a Gaussian integral. In momentum space, each mode $t$ couples only to the mode $t^{-1}$, just as $\phi_k$ couples only to $\phi_{-k}$. The average therefore becomes a product of small averages, each over a few modes.
\item \emph{Each small average is positive.} It contributes a factor $2^{-\rk\Upsilon(t)/2}$ for each momentum $t$ that it contains. A minus sign can occur only for modes that are their own ``complex conjugates.'' There it cancels, because $\Upsilon(t)$ is alternating and its rank is even.
\end{enumerate}
The only new feature is that all arithmetic is modulo two. The Fourier phases $e^{2\pi ik/n}$ must be replaced by $n$th roots of unity in a finite field, and there are enough of them only for odd $n$. This is the origin of the restriction to odd $n$.

\subsection{Fourier transform modulo two}

A Fourier transform on $n$ sites needs $n$ distinct $n$th roots of unity. The field $\F_2=\{0,1\}$ has only one, namely $1$. We therefore enlarge $\F_2$, just as one enlarges the real numbers to the complex numbers to find $e^{2\pi i/n}$. The finite field $\F_{2^m}$ with $2^m$ elements contains $\F_2$, and its nonzero elements form a cyclic group of order $2^m-1$. It therefore contains $n$ distinct $n$th roots of unity whenever $n$ divides $2^m-1$, and for odd $n$ we can always choose $m$ so that it does. For example, for $n=3$, the four-element field $\F_4=\{0,1,\omega,\omega^2\}$, with $\omega^2=\omega+1$, contains the cube root of unity $\omega$. We fix such a field and call it $\F$. For even $n$ the construction fails, because there are not enough distinct roots.

A replica momentum $t=(t_1,\ldots,t_{q-1})$ assigns an $n$th root of unity $t_a\in\F$ to each lattice direction. As in the main text, $t_q=1$ for the reference party. For a vertex $u$ of $G$, we write $t_u=t_{\alpha(u)}$ and $e_u=e_{\alpha(u)}$ for the momentum and the lattice direction of its party. The plane wave with momentum $t$ at the replica site $s=(s_1,\ldots,s_{q-1})$ is
\begin{equation}
\chi_t(s)=t_1^{s_1}\,t_2^{s_2}\cdots t_{q-1}^{s_{q-1}},
\end{equation}
which is the analogue of $e^{ik\cdot s}$. Plane waves are orthogonal for the usual reason: a nontrivial root of unity, summed over a full period, gives zero. Hence the replica Fourier transform
\begin{equation}
\hat x(t)=\sum_{s}\chi_t(s)^{-1}x_s,\qquad x_s=\sum_{t}\chi_t(s)\,\hat x(t),
\label{eq:replica-fourier}
\end{equation}
works as usual. Because $n=1$ modulo two, it does not even need a factor of $1/n$. It also has the familiar properties. A shift $x_s\to x_{s+c}$ multiplies $\hat x(t)$ by the phase $\chi_t(c)$, and Parseval's theorem reads
\begin{equation}
\sum_{s}x_s\,y_s=\sum_{t}\hat x(t)\,\hat y(t^{-1}),
\label{eq:parseval}
\end{equation}
with the inverse momentum $t^{-1}=(t_1^{-1},\ldots,t_{q-1}^{-1})$ in the role of $-k$.

\subsection{Momentum conservation and the rank sum}

The matrix $\Gamma$ commutes with the replica translations, so it cannot mix different momenta. In the sector with momentum $t$, it acts as an $N\times N$ matrix. We now show that this matrix is $\Upsilon(t)$, up to a rescaling that does not change the rank.

A vector $\mathbf x$ on the big graph has one bit $x_{s,u}$ for each copy $s$ and each vertex $u$ of $G$. We Fourier transform in $s$ separately for each $u$, and we collect the results in a vector $\hat x(t)$ with entries $\hat x_u(t)$. Each edge $\{u,v\}$ of $G$ gives two kinds of edges in $\mathcal G$ (the black and red edges in Fig.~\ref{fig:big_q3_n2}). An edge of the first kind joins $u$ and $v$ inside the copy $s$. An edge of the second kind comes from the permutations, and joins $u$ in copy $s$ to $v$ in copy $s+e_u-e_v$. Thus
\begin{equation}
(\Gamma\mathbf x)_{s,u}=\sum_v\gamma_{uv}\bigl(x_{s,v}+x_{s+e_u-e_v,\,v}\bigr).
\label{eq:Gamma-action}
\end{equation}
In momentum space, the shift in the second term becomes the phase $\chi_t(e_u-e_v)=t_u/t_v$. The edge therefore carries the weight
\begin{equation}
1+\frac{t_u}{t_v}=\frac{t_u-t_v}{t_v},
\end{equation}
where we used $1=-1$ modulo two. The factor $1/t_v$ only rescales column $v$. Apart from it, this is the entry of $\Upsilon(t)$ in Eq.~\eqref{eq:Upsilon-t}. The rank of $\Gamma$ is therefore $\sum_t\rk\Upsilon(t)$, which is Eq.~\eqref{eq:odd-replica-rank-sum}.

The weight $t_u-t_v$ has a simple meaning: an edge is invisible when its two endpoints carry the same momentum. In particular, an edge inside one party always drops out. In the big graph, the two edges that it produces then join the same pair of vertices, and they cancel modulo two. This is as it should be: a controlled-$Z$ gate inside one party is a local unitary, so it cannot change the R\'enyi multi-entropy.

\subsection{The Gauss sum is positive}
\label{app:A3}

The rank fixes only the magnitude of $Z_n^{(q)}$. For a general quadratic form, the average of $(-1)^{\Phi}$ is either zero or $\pm2^{-\rk\Gamma/2}$~\cite{Czech:2026cxw, Arf}. To find the sign, we evaluate $Z_n^{(q)}$ as one evaluates the Gaussian integral of a free lattice field. We write the exponent in momentum space, find the independent modes, and integrate over them one group at a time.

\paragraph*{The exponent in momentum space.}
If we list the edges of $\mathcal G$ copy by copy, as in Eq.~\eqref{eq:Gamma-action}, we get
\begin{equation}
\Phi_\mathcal G(\mathbf x)=\sum_{(u,v)}\,\sum_{s}x_{s,u}\bigl(x_{s,v}+x_{s+e_u-e_v,\,v}\bigr).
\end{equation}
With Parseval's theorem and the shift rule, this becomes
\begin{equation}
\Phi_\mathcal G(\mathbf x)=\sum_{\{t,\,t^{-1}\}}\hat x(t)^{\mathsf T}M(t)\,\hat x(t^{-1}).
\label{eq:phi-momentum}
\end{equation}
Here $M(t)$ is $\Upsilon(t)$ with row $u$ divided by $t_u$, and the sum is over the unordered pairs $\{t,t^{-1}\}$ other than zero momentum, which drops out. As for a free field, each mode couples only to the mode with the opposite momentum.

\paragraph*{Independent modes.}
The bits $x_{s,u}$ are independent, but their Fourier components are not. For a real field, $\phi_{-k}=\phi_k^*$, so only half of the modes are independent. The analogue here is
\begin{equation}
\hat x(t)^2=\hat x(t^2),
\label{eq:reality}
\end{equation}
which holds entry by entry, because a bit is equal to its own square and $(a+b)^2=a^2+b^2$ modulo two. Squaring thus plays the role of complex conjugation. Repeated squaring moves $t$ around an orbit $O=\{t,t^2,t^4,\ldots\}$, and Eq.~\eqref{eq:reality} fixes every mode on the orbit in terms of one of them, $\hat x(t_O)$. If the orbit has $k_O$ momenta, then $k_O$ squarings bring $\hat x(t_O)$ back to itself, so its entries lie in the subfield $\F_{2^{k_O}}$. These representative modes are the independent variables: an average over all bit strings $\mathbf x$ is the same as independent averages over the $\hat x(t_O)$.

\paragraph*{Factorization.}
A term of Eq.~\eqref{eq:phi-momentum} contains the modes on an orbit $O$ and on the inverse orbit $O^{-1}=\{t^{-1}:t\in O\}$. We collect the terms in groups, one for each set $O\cup O^{-1}$. Different groups contain different independent modes, so $Z_n^{(q)}$ is a product of averages, one for each group. The zero-momentum orbit $\{1\}$ does not appear, and it contributes $1$. We show below that every other group contributes
\begin{equation}
2^{-|O\cup O^{-1}|\,\varrho_O/2},
\label{eq:orbit-factor}
\end{equation}
which is a factor $2^{-\rk\Upsilon(t)/2}$ for each momentum $t$ in the group. Here $\varrho_O$ is the rank of $\Upsilon(t)$, which is the same for every $t$ in the group: replacing $t$ by $t^2$ squares the entries of $\Upsilon(t)$, replacing $t$ by $t^{-1}$ rescales its rows and columns, and neither changes the rank. The product over all groups is $2^{-\sum_t\rk\Upsilon(t)/2}=2^{-r_G(n)/2}$, which is the result.

There are two kinds of groups. They behave like the two familiar kinds of Gaussian integral: over two independent fields, and over one complex field.
\begin{itemize}
\item If $O^{-1}\neq O$, the modes $\hat x(t)$ and $\hat x(t^{-1})$ are independent. The group is like two independent fields $a$ and $b$ with the action $a^{\mathsf T}Mb$. The average over $b$ gives a delta function, which has no sign.
\item If $O^{-1}=O$, the mode $\hat x(t^{-1})$ is fixed by $\hat x(t)$: it is its ``complex conjugate.'' The group is like one complex field $a$ with a Hermitian action $a^\dagger Ma$. A sign can now occur, but it cancels because the rank of $\Upsilon(t)$ is even.
\end{itemize}

\paragraph*{The trace.}
The modes $\hat x(t)$ take values in the large field $\F$, but the exponent $\Phi_{\mathcal G}$ is a bit. The reality condition connects the two. Within a group, the terms of Eq.~\eqref{eq:phi-momentum} are successive squares of a single quantity $X$, which lies in a subfield $\F_{2^k}$. They therefore add up to
\begin{equation}
\operatorname{tr}_k(X)=X+X^2+X^4+\cdots+X^{2^{k-1}}.
\end{equation}
This field trace (not a matrix trace) plays the role of $X+X^*=2\operatorname{Re}X$, which turns a complex number into a real one. It is linear, and its value is always $0$ or $1$. Its key property is the analogue of $\int dz\,e^{i\operatorname{Re}(\bar cz)}\propto\delta(c)$:
\begin{equation}
\bigl\langle(-1)^{\operatorname{tr}_k(cz)}\bigr\rangle_z=
\begin{cases}1, & c=0,\\ 0, & c\neq0,\end{cases}
\label{eq:trace-sum}
\end{equation}
where $\langle\cdots\rangle_z$ is the average over all $z$ in $\F_{2^k}$.

\paragraph*{Groups with $O\neq O^{-1}$: two independent fields.}
Let $k$ be the number of elements of $O$, and let $M=M(t_O)$, $a=\hat x(t_O)$, and $b=\hat x(t_O^{-1})$. The group contributes $\operatorname{tr}_k(a^{\mathsf T}Mb)$ to the exponent. This is linear in $b$, so by Eq.~\eqref{eq:trace-sum} the average over $b$ is a delta function: it is $1$ if $a^{\mathsf T}M=0$ and $0$ otherwise. What remains is the fraction of vectors $a$ with $a^{\mathsf T}M=0$. These vectors form a subspace of codimension $\varrho_O$ over $\F_{2^k}$, so the fraction is $2^{-k\varrho_O}$. The group has $2k$ momenta, so this is Eq.~\eqref{eq:orbit-factor}. In continuum language, this is the integral $\int da\,db\,e^{ia^{\mathsf T}Mb}$, which is a delta function in disguise and has no sign.

\paragraph*{Groups with $O=O^{-1}$: one complex field.}
Now $t_O^{-1}$ is on the orbit of $t_O$, exactly halfway around: $t_O^{-1}=t_O^{\,Q}$, where $Q=2^j$ and the orbit has $2j$ elements. The map $z\mapsto\bar z\equiv z^Q$ is then a true complex conjugation. If we apply it twice we get the identity, and it leaves the subfield $\F_Q$ unchanged. The pair $\F_{Q^2}\supset\F_Q$ thus plays the role of $\mathbb C\supset\mathbb R$. Let $a=\hat x(t_O)$ and $M=M(t_O)$. Eq.~\eqref{eq:reality} gives $\hat x(t_O^{-1})=\bar a$, and similarly $\bar M=M^{\mathsf T}$, so $M$ is Hermitian. The group contributes $\operatorname{tr}_j(a^{\mathsf T}M\bar a)$ to the exponent, which is the analogue of a single complex field with a Hermitian action. By the lemma below, the average over $a$ is $(-1/Q)^{\varrho_O}$. The matrix $\Upsilon(t_O)$ is alternating, so $\varrho_O$ is even and the sign cancels. The average is therefore $Q^{-\varrho_O}=2^{-j\varrho_O}$. The group has $2j$ momenta, so this is again Eq.~\eqref{eq:orbit-factor}.

\begin{lemma}
\label{lem:hermitian-sum}
Let $M$ be an $N\times N$ Hermitian matrix over $\F_{Q^2}$, that is, $\bar M=M^{\mathsf T}$, and let $\varrho$ be its rank. Then the average of $(-1)^{\operatorname{tr}_j(a^{\mathsf T}M\bar a)}$ over all $a\in\F_{Q^2}^N$ is $(-1/Q)^{\varrho}$.
\end{lemma}

\begin{proof}[Sketch of proof]
As for an ordinary Hermitian matrix, a change of basis makes the form diagonal,\footnote{The Gram--Schmidt procedure needs a vector $a$ with $a^{\mathsf T}M\bar a\neq0$. A nonzero symmetric form over $\F_2$ need not have one, but a nonzero Hermitian form always does. Otherwise, expanding $(a+\lambda b)^{\mathsf T}M(\bar a+\bar\lambda\bar b)=0$ would give $z+z^Q=0$ with $z=\bar\lambda\,a^{\mathsf T}M\bar b$, for every $\lambda\in\F_{Q^2}$. This is impossible when $a^{\mathsf T}M\bar b\neq0$, because a polynomial of degree $Q$ cannot vanish at all $Q^2$ points.}
\begin{equation}
a^{\mathsf T}M\bar a=\sum_{i=1}^{\varrho}c_i\,a_i\bar a_i,\qquad 0\neq c_i\in\F_Q,
\end{equation}
and the average becomes a product of $\varrho$ averages over one variable each. For one variable $\lambda\in\F_{Q^2}$, we count the points on ``circles.'' The product $\lambda\bar\lambda=\lambda^{Q+1}$ is the analogue of $|\lambda|^2$. It is zero only for $\lambda=0$, and it takes each of the $Q-1$ nonzero values $w\in\F_Q$ exactly $Q+1$ times. By Eq.~\eqref{eq:trace-sum} for $\F_Q$, the nonzero values give $\sum_{w\neq0}(-1)^{\operatorname{tr}_j(cw)}=0-1=-1$. The average over $\lambda$ is therefore
\begin{equation}
\frac{1+(Q+1)(-1)}{Q^2}=-\frac1Q.
\end{equation}
\end{proof}

Each nonzero direction of a Hermitian form thus gives a factor $-1/Q$. This is the only possible source of a minus sign in the Gauss sum, and it cancels because the rank of $\Upsilon(t)$ is even. We conclude that $Z_n^{(q)}=2^{-r_G(n)/2}$ for every odd $n$: the Gauss sum is never zero and never negative. In the language of the main text, $\Phi_{\mathcal G}$ vanishes on the kernel of $\Gamma$, and the Arf invariant is zero.

\subsection{Local complementation}

Local complementation at a vertex $v$ toggles every edge between two neighbors of $v$. For the weighted matrix $\Upsilon(t)$, the same toggling is produced by a few row and column operations, and these cannot change the rank. This works because the weights are momentum differences, like potential differences in a circuit. The weight from $u$ to $w$ through $v$ is the same as the weight of a direct edge from $u$ to $w$:
\begin{equation}
(t_u-t_v)+(t_v-t_w)=t_u-t_w.
\end{equation}

\begin{proof}[Proof of Proposition~\ref{prop:lc-invariance}]
Fix a momentum assignment and let $A=\Upsilon_G(t)$. Add row $v$ of $A$ to the row of each neighbor of $v$. Then add column $v$ to the column of each neighbor of $v$. For two neighbors $u$ and $w$ of $v$, the row operation adds $A_{vw}=t_v-t_w$ to the entry $A_{uw}$, and the column operation adds $A_{uv}=t_u-t_v$. Together they add $t_u-t_w$, which is the weight of a direct edge from $u$ to $w$. If $\{u,w\}$ was not an edge, it now is. If it was an edge, the two equal weights cancel modulo two, and the edge is removed. No other entry changes. The result is $\Upsilon_{LC_v(G)}(t)$, so the two matrices have the same rank. This is true for every momentum assignment over every finite field of characteristic two. Local complementation therefore preserves collision control and the ranks $\varrho_G(\pi)$.
\end{proof}

\section{Derivation of Table~\ref{tab:results}}
\label{app:collision-examples}

In this appendix, each vertex is its own party. Thus $q=N$, and each vertex $u$ has its own momentum $t_u$.

\paragraph*{The idea.}
Think of the $n$ possible momenta as $n$ colors. A momentum assignment is then a coloring of the vertices of $G$, with no constraint between neighbors. We call an edge \emph{active} if its two endpoints have different colors. Only the active edges appear in $\Upsilon(t)$, because the weight $t_u-t_v$ is zero on the other edges. For each graph in Table~\ref{tab:results}, the rank of $\Upsilon(t)$ depends only on which edges are active. The rank sum is therefore a counting problem. We sort the colorings by their set of active edges, count each class, and add up the ranks. The multi-entropy is then the slope of the result at $n=1$.

The slope at $n=1$ has a simple meaning. At $n=1$ there is only one color, so no edge can be active. Near $n=1$, a class of colorings contributes only through the part of its size that is linear in $n-1$. In a tree, each active edge needs its own choice of a new color, which gives a factor $n-1$. So only the colorings with a single active edge contribute, and each edge gives one unit of $\log2$. In a cycle, the colors must also match up after one full turn. This constraint gives a term linear in $n-1$ for every set of active edges, with alternating signs. The net effect is that an odd cycle with $2r+1$ edges has $S_1=(r+1)\log2$, about half of the value that one unit per edge would give. Complete bipartite graphs have even cycles, but their $\Upsilon(t)$ is simple enough that we can find its rank directly.

\subsection{Summing over all colorings}

The main text fixes the momentum of the reference party to $t_q=1$. This is like fixing a gauge: if we multiply all momenta by the same root of unity, $\Upsilon(t)$ is multiplied by a constant, and its rank does not change. It is simpler to sum over all $n^N$ colorings. This counts each term of $r_G(n)$ exactly $n$ times:
\begin{equation}
R_G(n)\equiv\sum_{\text{all }t}\rk\Upsilon(t)=n\,r_G(n).
\label{eq:all-momenta}
\end{equation}
The extra factor of $n$ does not affect the limit $n\to1$, because the rank sum vanishes at $n=1$, and Eq.~\eqref{eq:rank-partition-function} gives
\begin{equation}
S_1(G)=\frac{\log2}{2}\,R_G'(1).
\label{eq:S1-from-R}
\end{equation}

\subsection{Graphs without even cycles}

For a graph without even cycles, Propositions~\ref{prop:matching} and~\ref{lem:no-even-cycle} give $\rk\Upsilon(t)=2\nu(F)$. Here $F$ is the set of active edges, and $\nu(F)$ is the size of the largest matching made of edges in $F$. Let $N_F(n)$ be the number of colorings whose active edges are exactly the edges in $F$. Then $R_G(n)=2\sum_FN_F(n)\,\nu(F)$, and Eq.~\eqref{eq:S1-from-R} gives
\begin{equation}
S_1(G)=\log2\sum_{F\subseteq E}N_F'(1)\,\nu(F).
\label{eq:S1-by-edges}
\end{equation}
Each set of edges $F$ contributes its matching number, weighted by the slope at $n=1$ of the number of colorings in which exactly these edges are active.

\subsection{Forests}

Let $G$ be a forest with $m$ edges. To count the colorings with active set $F$, give the root of each tree any of the $n$ colors, and move outward. Across an inactive edge the color is copied, which gives one choice. Across an active edge the color must change, which gives $n-1$ choices. Hence
\begin{equation}
N_F(n)=n^{\#\text{trees}}\,(n-1)^{|F|}.
\end{equation}
At $n=1$, this has slope one if $F$ is a single edge, and slope zero if $F$ has two or more edges. A single edge has $\nu=1$, so each edge contributes one unit of $\log2$:
\begin{equation}
S_1(G)=m\log2.
\end{equation}

\subsection{Odd cycles}

Let $G=C_{2r+1}$. To count the colorings with active set $F$, merge the two endpoints of each inactive edge, since they have the same color. This leaves a shorter cycle with $|F|$ vertices, in which neighbors have different colors. We count these colorings with a transfer matrix. One step around the cycle must change the color, so it is described by the $n\times n$ matrix $J-I$, where $J$ is the matrix of all ones. Closing the cycle gives a trace. The eigenvalues of $J-I$ are $n-1$ (once) and $-1$ ($n-1$ times), so
\begin{equation}
N_F(n)=\Tr(J-I)^{|F|}=(n-1)^{|F|}+(-1)^{|F|}(n-1).
\end{equation}
This is the chromatic polynomial of a cycle. The first term is the same as for a tree. The second term comes from closing the cycle, and it is linear in $n-1$ for every $F$. It also makes $N_F=0$ for a single active edge: the color cannot change exactly once around the cycle and still return to its initial value. So no single edge contributes, but every set of two or more active edges does, with slope $N_F'(1)=(-1)^{|F|}$. Eq.~\eqref{eq:S1-by-edges} thus becomes an alternating sum,
\begin{equation}
\frac{S_1(C_{2r+1})}{\log2}=\sum_{|F|\geq2}(-1)^{|F|}\,\nu(F).
\label{eq:cycle-alternating}
\end{equation}
These alternating signs are what allow the cycle to have a smaller continued value than the path that remains when one of its vertices is deleted (see the main text).

To evaluate the sum, we first include the empty set and the single edges, and we remove them again at the end. The active edges of $F$ form \emph{runs}, which are maximal blocks of consecutive active edges. A run of $\ell$ edges is a path. Its maximum matching takes every other edge, so it has $\lceil\ell/2\rceil$ edges. Hence $2\nu(F)=|F|+o(F)$, where $o(F)$ is the number of runs of odd length. The only exception is the whole cycle $E$, for which $2\nu(E)=2r=|E|-1$ (we set $o(E)=0$). Thus
\begin{equation}
2\sum_{F}(-1)^{|F|}\nu(F)=\sum_F(-1)^{|F|}|F|+\sum_F(-1)^{|F|}o(F)+1,
\end{equation}
where the final $+1$ is the correction for $F=E$.

Most terms of these alternating sums cancel in pairs. If an edge can be switched on or off without a change in the summand, the two choices have opposite signs and cancel. Only the configurations without such a free edge remain.
\begin{itemize}
\item The sum of $(-1)^{|F|}|F|$ counts, for each edge $e$, the sets $F$ that contain $e$. Every other edge is free, so this sum is zero.
\item The sum of $(-1)^{|F|}o(F)$ counts, for each run $R$ of odd length, the sets $F$ in which $R$ is a run. For this, $R$ must be active and the two edges next to $R$ must be inactive, and all other edges are free. So a run contributes only if there are no other edges, that is, if it has $2r-1$ edges and $F=R$. There are $2r+1$ such runs, each with sign $-1$, so this sum is $-(2r+1)$.
\end{itemize}
The sum over all $F$ is therefore $\frac12[0-(2r+1)+1]=-r$. The empty set contributes nothing, and each of the $2r+1$ single edges contributes $-1$. When we remove them, we get
\begin{equation}
S_1(C_{2r+1})=(-r+2r+1)\log2=(r+1)\log2.
\end{equation}
As a check, the triangle ($r=1$) gives $2\log2$. This is the value for the three-party GHZ state, as it must be. The triangle is obtained by local complementation at the middle vertex of the three-vertex path, which is the star graph of Fig.~\ref{fig:ghz-star-graph} for $q=3$.

\end{document}